\documentclass[runningheads]{llncs}
\usepackage{amsmath}
\usepackage{amsfonts}
\usepackage{graphicx}
\usepackage[frozencache,cachedir=.]{minted}
\usepackage{caption}
\usepackage{subcaption}
\usepackage{multirow}
\usepackage{scalerel}
\usepackage{algorithm}
\usepackage[noend]{algpseudocode}
\algnewcommand{\LineComment}[1]{\State \(\triangleright\) #1}

\usepackage{tikz}

\newcounter{stackCounter}

\newcommand{\reverseList}[1]{
  \let\reversedList\empty
  \foreach\x in {#1} {
    \ifx\reversedList\empty
      \xdef\reversedList{\x}%
    \else
      \xdef\reversedList{\x,\reversedList}%
    \fi
  }
}

\newcommand{\stack}[2][1]{%
  \tikzstyle{top}=[ultra thin, fill=brown!50]%
  \tikzstyle{bottom}=[ultra thin, fill=brown!58]%
  \pgfmathsetmacro\pancakeHeight{0.1}%
  \pgfmathsetmacro\pancakeWidthMin{0.2}%
  \pgfmathsetmacro\pancakeWidthStep{0.05}%
  \pgfmathsetmacro\pancakeThickness{0.025}%
  \pgfmathsetmacro\stackStep{0.1}%
  \begin{tikzpicture}%
    \begin{scope}[scale=#1]%
      \setcounter{stackCounter}{0}%
      \reverseList{#2}%
      \foreach \x in \reversedList {%
        \pgfmathsetmacro\yBottom{\thestackCounter * \stackStep}%
        \pgfmathsetmacro\yTop{\thestackCounter * \stackStep + \pancakeThickness}%
        \pgfmathabs{\x}%
        \pgfmathsetmacro\xWidth{\pancakeWidthMin + (\pgfmathresult-1) * \pancakeWidthStep}%
        \draw[style=bottom] (0,\yBottom) ellipse ({\xWidth} and {\pancakeHeight});
        \draw[style=top] (0,\yTop) ellipse ({\xWidth} and {\pancakeHeight});
        \stepcounter{stackCounter}%
      }%
    \end{scope}%
  \end{tikzpicture}%
}%

\usepackage{tikz}
\usepackage{xstring}

\newcounter{marbleCounter}

\newcommand{\AlgGreedy}[1]{\mathsf{Twisted}(#1)}
\newcommand{\AlgRulerU}[1]{\mathsf{Ruler}(#1)}
\newcommand{\AlgRulerS}[1]{\mathsf{Ruler\pm}(#1)}
\newcommand{\AlgGrayU}[1]{\mathsf{GrayCode}(#1)}
\newcommand{\AlgGrayS}[1]{\mathsf{GrayCode\pm}(#1)}
\newcommand{\VisitName}{\mathsf{visit}}
\newcommand{\Visit}[1]{\VisitName(#1)}
\newcommand{\Yield}[1]{\mathsf{yield}(#1)}

\graphicspath{{graphics/}{operations/}}

\usepackage{tikz}

\definecolor{p1colPlus1}{rgb}{0.34, 0.62, 1.00}
\definecolor{p1colPlus2}{rgb}{0.63, 0.92, 0.55}
\definecolor{p1colPlus3}{rgb}{1.00, 0.89, 0.30}
\definecolor{p1colPlus4}{rgb}{1.00, 0.16, 0.16}
\definecolor{p1col+1}{rgb}{0.34, 0.62, 1.00}
\definecolor{p1col+2}{rgb}{0.63, 0.92, 0.55}
\definecolor{p1col+3}{rgb}{1.00, 0.89, 0.30}
\definecolor{p1col+4}{rgb}{1.00, 0.16, 0.16}
\definecolor{p1col1}{rgb}{0.27, 0.49, 0.80}
\definecolor{p1col2}{rgb}{0.51, 0.73, 0.44}
\definecolor{p1col3}{rgb}{0.86, 0.67, 0.24}
\definecolor{p1col4}{rgb}{0.86, 0.13, 0.13}
\definecolor{p1colMinus1}{rgb}{0.20, 0.36, 0.60}
\definecolor{p1colMinus2}{rgb}{0.39, 0.54, 0.33}
\definecolor{p1colMinus3}{rgb}{0.64, 0.45, 0.18}
\definecolor{p1colMinus4}{rgb}{0.64, 0.10, 0.10}
\definecolor{p1col-1}{rgb}{0.20, 0.36, 0.60}
\definecolor{p1col-2}{rgb}{0.39, 0.54, 0.33}
\definecolor{p1col-3}{rgb}{0.64, 0.45, 0.18}
\definecolor{p1col-4}{rgb}{0.64, 0.10, 0.10}

\definecolor{p2col1}{rgb}{0.264135, 0.201429, 0.745889}
\definecolor{p2col2}{rgb}{0.256326, 0.430921, 0.808553}
\definecolor{p2col3}{rgb}{0.324106, 0.60897, 0.708341}
\definecolor{p2col4}{rgb}{0.439128, 0.704968, 0.52925}
\definecolor{p2col5}{rgb}{0.597181, 0.742185, 0.36771}
\definecolor{p2col6}{rgb}{0.764712, 0.728302, 0.273608}
\definecolor{p2col7}{rgb}{0.88018, 0.631684, 0.227665}
\definecolor{p2col8}{rgb}{0.897354, 0.41824, 0.185007}
\definecolor{p2col9}{rgb}{0.857359, 0.131106, 0.132128}

\newcommand{\drawUnsigned}[1]{
  \xdef\symbolWidth{1}%
  \xdef\symbolHeightMin{0.2}%
  \xdef\symbolHeightInc{0.2}%
  \xdef\symbolSpacing{0.0}%

  \pgfmathparse{dim(#1)}%
  \xdef\totalSymbols{\pgfmathresult}%

  \begin{tikzpicture}%
    \begin{scope}[scale=0.2]%
      \pgfmathparse{int(\totalSymbols-1)}%
      \xdef\lastIndex{\pgfmathresult}
      \foreach \arrayIndex in {0,1,...,\lastIndex}{%
        \pgfmathparse{{#1}[\arrayIndex]}%
        \xdef\arraySymbol{\pgfmathresult}%
        \pgfmathparse{\symbolWidth*\arrayIndex+\symbolSpacing*(\arrayIndex-1)}%
        \xdef\xLeft{\pgfmathresult}%
        \pgfmathparse{\xLeft+\symbolWidth}%
        \xdef\xRight{\pgfmathresult}%
        \pgfmathparse{0}%
        \xdef\yBottom{\pgfmathresult}%
        \pgfmathparse{\yBottom+\symbolHeightMin+\symbolHeightInc*(\arraySymbol-1)}%
        \xdef\yTop{\pgfmathresult}%
        \draw[thin, fill=p1col\arraySymbol] (\xLeft,\yBottom) rectangle (\xRight,\yTop);%
      }%
    \end{scope}%
  \end{tikzpicture}%
}

\newcommand{\drawSigned}[1]{
  \xdef\symbolWidth{1}%
  \xdef\symbolHeightMin{0.4}%
  \xdef\symbolHeightInc{0.1}%
  \xdef\symbolSpacing{0.0}%

  \pgfmathparse{dim(#1)}%
  \xdef\totalSymbols{\pgfmathresult}%

  \begin{tikzpicture}%
    \begin{scope}[scale=0.2]%
      \pgfmathparse{int(\totalSymbols-1)}%
      \xdef\lastIndex{\pgfmathresult}
      \foreach \arrayIndex in {0,1,...,\lastIndex}{%
        \pgfmathparse{int({#1}[\arrayIndex])}%
        \xdef\signedSymbol{\pgfmathresult}%
        \pgfmathparse{int(abs(\signedSymbol))}%
        \xdef\unsignedSymbol{\pgfmathresult}%
        \pgfmathparse{\symbolWidth*\arrayIndex+\symbolSpacing*(\arrayIndex-1)}%
        \xdef\xLeft{\pgfmathresult}%
        \pgfmathparse{\xLeft+\symbolWidth}%
        \xdef\xRight{\pgfmathresult}%
        \pgfmathparse{0}%
        \xdef\yBottom{\pgfmathresult}%
        \ifnum \signedSymbol>0%
          \pgfmathparse{\yBottom+\symbolHeightMin+\symbolHeightInc*(\unsignedSymbol-1)}%
        \else%
          \pgfmathparse{(\yBottom+\symbolHeightMin+\symbolHeightInc*(\unsignedSymbol-1))}%
        \fi%
        \xdef\yTop{\pgfmathresult}%
        \draw[thin, fill=p1col\signedSymbol] (\xLeft,\yBottom) rectangle (\xRight,\yTop);%
      }%
    \end{scope}%
  \end{tikzpicture}%
}

\newcommand{\drawSignedPositions}[2]{
  \xdef\symbolWidth{1}%
  \xdef\symbolHeightMin{0.2}%
  \xdef\symbolHeightInc{0.2}%
  \xdef\symbolSpacing{0.0}%

  \pgfmathparse{dim(#1)}%
  \xdef\totalSymbols{\pgfmathresult}%

  \begin{tikzpicture}%
    \begin{scope}[scale=0.2]%
      \pgfmathparse{int(\totalSymbols-1)}%
      \xdef\lastIndex{\pgfmathresult}
      \foreach \arrayIndex in {0,1,...,\lastIndex}{%
        \pgfmathparse{{#1}[\arrayIndex]}%
        \xdef\arraySymbol{\pgfmathresult}%
        \pgfmathparse{\symbolWidth*\arrayIndex+\symbolSpacing*(\arrayIndex-1)}%
        \xdef\xLeft{\pgfmathresult}%
        \pgfmathparse{\xLeft+\symbolWidth}%
        \xdef\xRight{\pgfmathresult}%
        \pgfmathparse{0}%
        \xdef\yBottom{\pgfmathresult}%
        \pgfmathparse{{#2}[\arrayIndex]}%
        \xdef\bit{\pgfmathresult}%
        \ifnum \bit=0%
          \pgfmathparse{\yBottom+\symbolHeightMin+\symbolHeightInc*(\arraySymbol-1)}%
        \else%
          \pgfmathparse{-(\yBottom+\symbolHeightMin+\symbolHeightInc*(\arraySymbol-1))}%
        \fi%
        \xdef\yTop{\pgfmathresult}%
        \draw[thin, fill=p1col\arraySymbol] (\xLeft,\yBottom) rectangle (\xRight,\yTop);%
      }%
    \end{scope}%
  \end{tikzpicture}%
}

\newcommand{\BINARY}[1]{B_{#1}}
\newcommand{\PERMS}[1]{S_{#1}}
\newcommand{\SIGNED}[1]{S_{#1}^{\pm}}

\newcommand{\swapl}[1]{\overleftarrow{#1}}
\newcommand{\swapr}[1]{\overrightarrow{#1}}
\newcommand{\inc}[1]{\overline{#1}}
\newcommand{\dec}[1]{\underline{#1}}

\newcommand{\twtLeft}[1]{\overleftarrow{\mathsf{t}}_{\!\!#1}}
\newcommand{\twtRight}[1]{\overrightarrow{\mathsf{t}}_{\!\!#1}}
\newcommand{\twt}[1]{\mathsf{t}_{#1}}

\newcommand{\plain}[1]{\mathsf{plain}(#1)}
\newcommand{\twisted}[1]{\mathsf{twisted}(#1)}

\newcommand{\bigO}[1]{\mathcal{O}(#1)}

\newcommand{\rulerName}{\mathsf{ruler}}
\newcommand{\ruler}[1]{\rulerName(#1)}
\newcommand{\srulerName}{\mathsf{ruler\pm}}
\newcommand{\sruler}[1]{\srulerName(#1)}

\DeclareMathSymbol{\shortminus}{\mathbin}{AMSa}{"39}

\newcommand{\colorOne}[1]{\textbf{\textcolor{p1col1}{#1}}}
\newcommand{\colorTwo}[1]{\textbf{\textcolor{p1col2}{#1}}}
\newcommand{\colorThree}[1]{\textbf{\textcolor{p1col3}{#1}}}
\newcommand{\colorFour}[1]{\textbf{\textcolor{p1col4}{#1}}}

\newcommand{\colorMinusTwo}[1]{\textbf{\textcolor{p1col-2}{#1}}}
\newcommand{\colorMinusThree}[1]{\textbf{\textcolor{p1col-3}{#1}}}
\newcommand{\colorMinusFour}[1]{\textbf{\textcolor{p1col-4}{#1}}}

\newcommand{\colorThr}[1]{\colorThree{#1}}
\newcommand{\colorFou}[1]{\colorFour{#1}}

\begin{document}
%
\title{Generating Signed Permutations by \\ Twisting Two-Sided Ribbons}
%
\titlerunning{Generating Signed Permutations by Twisting Two-Sided Ribbons}
%

\author{Yuan (Friedrich) Qiu \inst{1} \and
Aaron Williams\inst{1}\orcidID{0000-0001-6816-4368}}
\authorrunning{Y. Qiu \and A. Williams}
%
\institute{Williams College, Williamstown MA 01267, USA \\
\url{https://csci.williams.edu/people/faculty/aaron-williams/} \\
\email{\{yq1, aaron.williams\}@williams.edu}}
%

%

\maketitle              
\begin{abstract}
We provide a simple and natural solution to the problem of generating all $2^n \cdot n!$ signed permutations of $[n] = \{1,2,\ldots,n\}$.
Our solution provides a pleasing generalization of the most famous ordering of permutations: plain changes (Steinhaus-Johnson-Trotter algorithm).
In plain changes, the $n!$ permutations of $[n]$ are ordered so that successive permutations differ by swapping a pair of adjacent symbols, and the order is often visualized as a weaving pattern involving $n$ ropes.
Here we model a signed permutation using $n$ ribbons with two distinct sides, and each successive configuration is created by twisting (i.e., swapping and turning over) two neighboring ribbons or a single ribbon.
By greedily prioritizing $2$-twists of the largest symbol before $1$-twists of the largest symbol, we create a signed version of plain change's memorable zig-zag pattern. 
We provide a loopless algorithm (i.e., worst-case $\bigO{1}$-time per object) by extending the well-known mixed-radix Gray code algorithm.

\end{abstract}

\section{Introduction} 
\label{sec:intro}

This section discusses the problem of generating permutations, and our goal of creating a memorable solution to the problem of generating signed permutations.

\subsection{Generating Permutations}
\label{sec:intro_perms}

The generation of permutations is a classic problem that dates back to the dawn of computer science (and several hundred years earlier).
The goal is to create all $n!$ permutations of $[n] = \{1,2,\ldots,n\}$ as efficiently as possible.
A wide variety of approaches have been considered, some of which can be conceptualized using a specific physical model of the permutation.
We now consider three such examples.

Zaks' algorithm \cite{zaks1984new} can be conceptualized using a stack of $n$ pancakes of varying sizes.
Successive permutations are created by flipping some pancakes at the top of the stack, which is equivalent to a \emph{prefix-reversal} in the permutation.
For example, if \raisebox{-0.15em}{\stack[0.75]{1,2,3,4}} represents $1234$, then flipping the top three pancakes gives \raisebox{-0.15em}{\stack[0.75]{3,2,1,4}} or $\overleftrightarrow{123}4 = 3214$. 
Table \ref{tab:orders} shows the full order for $n=4$.
Zaks' designed his `new' order to have an efficient array-based implementation.
Unknown to Zaks, Kl\"{u}gel had discovered the order by 1796 \cite{hindenburg1796sammlung}; see \cite{cameron2023hamiltonicity} for further details.

Corbett's algorithm \cite{corbett1992rotator} can be conceptualized using $n$ marbles on a ramp.
Successive permutations are created by moving a marble to the top of the ramp, which is equivalent to a \emph{prefix-rotation} in the permutation.
For example, if \raisebox{-0.15em}{\includegraphics[height=1.25em]{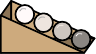}} represents $1234$, then moving the fourth marble gives \raisebox{-0.15em}{\includegraphics[height=1.25em]{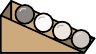}} or $\overleftarrow{1234} = 4123$. 

The algorithms by Zaks and Corbett are reasonably well-known, and have their own specific applications.
For example, in interconnection networks \cite{duato2003interconnection}, the algorithms provide Hamiltion cycles in the pancake network and rotator network, respectively.
However, neither can lay claim to being \underline{the} permutation algorithm.

\emph{Plain changes} can be conceptualized using $n$ ropes running in parallel.
Successive permutations are obtained by crossing one rope over a neighboring rope, as in weaving, which is equivalent to a \emph{swap} (or \emph{adjacent-transposition}) in the permutation.
For example, if \raisebox{0em}{\includegraphics[height=0.75em]{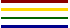}} represents $1234$, then swapping the middle two ropes gives \raisebox{0em}{\includegraphics[height=0.75em]{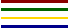}} or $1\overleftrightarrow{23}4 = 1324$.
Plain changes dates back to bell-ringers in the 1600s \cite{duckworth1668tintinnalogia,stedman1677campanalogia}.
Figure \ref{fig:plain4} shows the full order for $n=4$, including its familiar zig-zag pattern.
It is commonly known as the \emph{Steinhaus-Johnson-Trotter algorithm} \cite{steinhaus1979one} \cite{johnson1963generation} \cite{trotter1962algorithm} due to repeated rediscoveries circa 1960.

Many other notable approaches to permutation generation exist, with surveys by Sedgewick \cite{sedgewick1977permutation}, Savage \cite{savage1997survey}, and M\"{u}tze \cite{mutze2022combinatorial}, frameworks by Knuth \cite{knuth2013art} and Ganapathi and Chowdhury \cite{Ganapathi23}.
While some methods have specific advantages \cite{holroyd2012shorthand} or require less additional memory when implemented \cite{liptak23constant}, there is little doubt that plain changes is \underline{the} solution to permutation generation.

\subsection{Generating Signed Permutations}
\label{sec:intro_signed}

A \emph{signed permutation} of $[n]$ is a permutation of $[n]$ in which every symbol is given a $\pm$ sign.
We let $\PERMS{n}$ and $\SIGNED{n}$ be the sets of all permutations and signed permutations of $[n]$, respectively.
Note that $|\PERMS{n}| = n!$ and $|\SIGNED{n}| = 2^n \cdot n!$.
For example, $231 \in \PERMS{3}$ has eight different signings, including ${+}2{-}3{-}1 \in \SIGNED{3}$.
For convenience, we sometimes use overlines for negative symbols, with $2\overline{3}\overline{1}$ denoting ${+}2{-}3{-}1$.
Signed permutations arise in many natural settings, including genomics \cite{fertin2009combinatorics}.

While the permutation generation problem has a widely accepted best answer, the same is not true for signed permutations.
A solution by Korsh, LaFollette, and Lipschutz \cite{korsh2011loopless} achieves the optimal running-time (see Section \ref{sec:combgen}), however, it operates on the signed permutation as a composite object.
In other words, it either modifies the underlying permutation, or the signs, but not both at the same time.
More specifically, their solution either swaps two symbols (and preserves their signs) or changes the rightmost symbol's sign.
An alternate solution by Suzuki, Sawada, and Kaneko \cite{suzuki2005hamilton} treats signed permutations as stacks of $n$ \emph{burnt pancakes}, and can be understood as a signed version of Zaks' algorithm.


\subsubsection{Physical Model: Ribbons}
\label{sec:intro_signed_physical}

In keeping with the previously discussed permutation generation algorithms, we will provide a physical model for signed permutations.
A \emph{two-sided ribbon} is glossy on one side and matte on the other%
\footnote{Manufacturers refer to this type of ribbon as \emph{single face} as only one side is polished.},
and we consider $n$ such ribbons running in parallel.
We modify the ribbons by \emph{twisting} some number of neighboring ribbons.
More specifically, a \emph{$k$-twist} turns over $k$ consecutive ribbons and reverses their order, as shown in Figure \ref{fig:operations}.
In other contexts, a twist is a \emph{complementing substring reversal}, or simply a \emph{reversal} \cite{HannenhalliP95}.

\begin{figure}[h]
    \vspace{-1.5em}
    \centering
    \begin{subfigure}{0.485\textwidth}
        \centering
        \begin{tabular}{rcr}
        1 &
        \multirow{4}{*}{\raisebox{0.5em}{\includegraphics[height=0.5in]{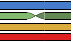}}} &
        \colorOne{1} \\[-0.2em]
        \colorTwo{2} && \colorMinusTwo{-2} \\[-0.2em]
        \colorThree{3} && \colorThree{3} \\[-0.2em]
        \colorFour{4} && \colorFour{4} \\[-0.2em]
        \end{tabular}
        \caption{The $1$-twist changes $1 \, 2 \, 3 \, 4$ into $1 \, \overline{2} \, 3 \, 4$.}
        \label{fig:operations_1twist}
    \end{subfigure}
    \hfill
    \begin{subfigure}{0.485\textwidth}
        \centering
        \begin{tabular}{rcr}
        1 &
        \multirow{4}{*}{\raisebox{0.5em}{\includegraphics[height=0.5in]{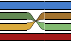}}} &
        \colorOne{1} \\[-0.2em]
        \colorMinusTwo{-2} && \colorMinusThree{-3} \\[-0.2em]
        \colorThree{3} && \colorTwo{2} \\[-0.2em]
        \colorMinusFour{-4} && \colorMinusFour{4} \\[-0.2em]
        \end{tabular}
        \caption{The $2$-twist changes $1 \, \overline{2} \, 3 \, \overline{4}$ into $1 \, \overline{3} \, 2 \, \overline{4}$.}
        \label{fig:operations_2twist}
    \end{subfigure}
    \vspace{-0.5em}
    \caption{Two-sided ribbons with distinct positive (i.e., glossy) and negative (i.e., matte) sides running in parallel.
    A $k$-twist reverses the order of $k$ neighboring ribbons and turns each of them over, as shown for (a) $k=1$ and (b) $k=2$.}
    \label{fig:operations}
\end{figure}


Our goal is to create a \emph{twist Gray code} for signed permutations.
This means that successive entries of $\SIGNED{n}$ are always created by applying a single twist, or equivalently, a sequence of $2^n n! - 1$ twists generate each entry of $\SIGNED{n}$ in turn.
It should be obvious that $1$-twists are insufficient for this task on their own, as they do not modify the underlying permutation.
Similarly, $2$-twists are insufficient on their own, as they do not modify the number of positive symbols modulo two.
However, we will show that $2$-twists and $1$-twists are sufficient when used together.
Moreover, we provide a simple solution using only these shortest twists.

\subsection{Outline}
\label{sec:intro_outline}

Section~\ref{sec:combgen} provides background on combinatorial generation, including greedy Gray codes.
Section~\ref{sec:signedPlain} presents our twist Gray code, which can be understood as greedily prioritizing $2$-twists of the largest possible symbol followed by $1$-twists of the the largest possible symbol.
We name our order \emph{twisted plain changes}, and we propose it to be \underline{the} solution to the signed permutation generation problem.
Section~\ref{sec:ruler} discusses ruler sequences, and Section~\ref{sec:loopless} provides a loopless implementation of our algorithm.
A Python implementation of our loopless algorithm appears in the appendix.

\section{Combinatorial Generation}
\label{sec:combgen}



As Ruskey explains in his \emph{Combinatorial Generation} manuscript \cite{ruskey2003combinatorial}, humans have been making exhaustive lists of various kinds for thousands of years, and more recently, they have been programming computers to complete the task.
The latter pursuit requires the creation of a suitable order of the objects to be generated, and the design of an algorithm to efficiently generate the objects in the desired order.
In this section we review some basic concepts and terminology, and then discuss two specific foundational results.

\subsection{Gray Codes and Loopless Algorithms}
\label{sec:combgen_Gray}

If successive objects in an order differ in a constant amount (by a natural metric), then the order can be called a \emph{Gray code}.
If an algorithm generates successive objects in $\bigO{1}$-time in an amortized or worst-case sense, then it runs in \emph{constant amortized time} (CAT) or is \emph{loopless} (i.e., there are no internal loops) \cite{Ehrlich73}, respectively.
To make sense of these definitions, note that a well-written generation algorithm shares one copy of an object with its applications.
The generation algorithm then repeatedly modifies the object and informs the application that the `next' object is available to visit, without returning the full object.
This \emph{shared object model} allows the generation algorithm to create successive objects in constant time, and loopless algorithms require a Gray code ordering.

For example, Zaks' order can be generated by a CAT algorithm when the permutation is stored in an array.
This is due to the fact that the prefix-reversals that are performed by the algorithm have a constant length on average (see Ord-Smith's earlier \textsc{EconoPerm} implementation \cite{ord1970generation}).
However, a loopless implementation is not possible, since prefix-reversals of length $n$ (i.e., flipping the entire stack) takes $\Theta(n)$-time%
\footnote{If the permutation is stored in a BLL, then a loopless implementation is possible~\cite{williams20101}.}.
On the other hand, plain changes can be implemented by a loopless algorithm, as we will see in Sections \ref{sec:ruler}--\ref{sec:loopless}.

\subsection{The Greedy Gray Code Algorithm}
\label{sec:combgen_greedy}

Historically, most Gray codes have been constructed recursively.
In other words, the order of larger objects is based on the order of smaller objects.
We prefer a more recent approach that has provided simpler interpretations of previous constructions, and the discovery of new orders (including twisted plain changes).

The \emph{greedy Gray code algorithm} \cite{williams2013greedy} attempts to create a Gray code one object at time, based on a starting object, and a prioritized list of modifications.
At each step, it tries to extend the order to a new object by applying a modification to the most recently created object.
If this is possible, then the highest priority modification is used.
All three permutation algorithms discussed in Section \ref{sec:intro_perms} have simple greedy descriptions \cite{williams2013greedy}.
Furthermore, the aforementioned burnt pancake order of signed permutations \cite{suzuki2005hamilton} has a greedy interpretation:
flip the fewest burnt pancakes (or equivalently, apply the shortest \emph{prefix twist}) \cite{sawada2016greedy}.

The greedy approach is clarified using the ``original'' Gray code in Section~\ref{sec:combgen_BRGC}.




\subsection{Binary Reflected Gray Code}
\label{sec:combgen_BRGC}

Amongst all results in combinatorial generation, plain change's stature is rivaled only by the \emph{binary reflected Gray code}%
\footnote{The eponymous \emph{Gray code} by Gray \cite{gray1953pulse} also demonstrates Stigler's law \cite{stigler1980stigler}: \cite{gardner1972curious}~\cite{heath1972origins}.}.
The BRGC orders $n$-bit binary strings by \emph{bit-flips}, meaning successive strings differ in one bit.
It is generated starting at the all-$0$s string by greedily flipping the rightmost possible bit \cite{williams2013greedy}.
For example, the order for $n=4$ begins as follows, where overlines denote the bit that is flipped to create the next binary string,
\begin{equation} \label{eq:BRGC4}
000\overline{0}, 00\overline{0}1, 001\overline{1}, 0010, \ldots.
\end{equation}
To continue \eqref{eq:BRGC4} we consider bit-flips in $0010$ from right to left.
We can't flip the right bit since $001\overline{0} = 0011$ is already in the list.
Similarly, $00\overline{1}0 = 0000$ is also in the list.
However, $0\overline{0}10 = 0110$ is not in the list, so this becomes the next string.

The full BRGC order for $n=4$ is visualized in Figure \ref{fig:brgc4} using two-sided ribbons, where each bit-flip is a $1$-twist of the corresponding ribbon.
Note that the order is \emph{cyclic}, as the last and first strings differ by flipping the first bit.

\begin{figure}[h]
    \centering
    \footnotesize
    \includegraphics[angle=270,width=\textwidth]{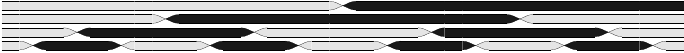}
    \begin{tabular}{@{}*{15}{c@{\hspace{1.95em}}}c@{}}
0 & 0 & 0 & 0 & 0 & 0 & 0 & 0 & 1 & 1 & 1 & 1 & 1 & 1 & 1 & 1 \\[-0.25em]
0 & 0 & 0 & 0 & 1 & 1 & 1 & 1 & 1 & 1 & 1 & 1 & 0 & 0 & 0 & 0 \\[-0.25em]
0 & 0 & 1 & 1 & 1 & 1 & 0 & 0 & 0 & 0 & 1 & 1 & 1 & 1 & 0 & 0 \\[-0.25em]
0 & 1 & 1 & 0 & 0 & 1 & 1 & 0 & 0 & 1 & 1 & 0 & 0 & 1 & 1 & 0 \\[-0.5em]
    \end{tabular}
    \caption{Binary reflected Gray code using indistinct two-sided~ribbons for $n=4$.
    }
    \label{fig:brgc4}
\end{figure}

\subsection{Plain Changes}
\label{sec:combgen_plain}

Plain changes can be defined greedily: start at $12 \cdots n$ and then swap the largest possible symbol \cite{williams2013greedy}.
This may seem underspecified---should we swap a symbol to the left or the right?---but the swapped symbol is always in the leftmost or rightmost position, or the opposite swap recreates a previous permutation.

The order $\plain{n}$ for $n=4$ is visualized in Figure \ref{fig:plain4} using distinct one-sided ribbons%
\footnote{Physically, a ribbon moves above or below its neighbor, but that is not relevant~here.}. 
The characteristic zig-zag pattern of the largest symbol \colorFour{4} should be apparent.
The order is \emph{cyclic} by swapping the first two symbols.

\begin{figure}[h]
    \centering
     \footnotesize
    \includegraphics[angle=270,width=\textwidth]{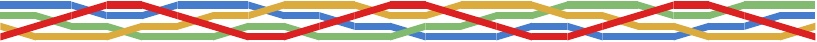}
    \begin{tabular}{@{}*{23}{c@{\hspace{1.03em}}}c@{}}
\colorOne{1} & \colorOne{1} & \colorOne{1} & \colorFou{4} & \colorFou{4} & \colorOne{1} & \colorOne{1} & \colorOne{1} & \colorThr{3} & \colorThr{3} & \colorThr{3} & \colorFou{4} & \colorFou{4} & \colorThr{3} & \colorThr{3} & \colorThr{3} & \colorTwo{2} & \colorTwo{2} & \colorTwo{2} & \colorFou{4} & \colorFou{4} & \colorTwo{2} & \colorTwo{2} & \colorTwo{2} \\[-0.25em]
\colorTwo{2} & \colorTwo{2} & \colorFou{4} & \colorOne{1} & \colorOne{1} & \colorFou{4} & \colorThr{3} & \colorThr{3} & \colorOne{1} & \colorOne{1} & \colorFou{4} & \colorThr{3} & \colorThr{3} & \colorFou{4} & \colorTwo{2} & \colorTwo{2} & \colorThr{3} & \colorThr{3} & \colorFou{4} & \colorTwo{2} & \colorTwo{2} & \colorFou{4} & \colorOne{1} & \colorOne{1} \\[-0.25em]
\colorThr{3} & \colorFou{4} & \colorTwo{2} & \colorTwo{2} & \colorThr{3} & \colorThr{3} & \colorFou{4} & \colorTwo{2} & \colorTwo{2} & \colorFou{4} & \colorOne{1} & \colorOne{1} & \colorTwo{2} & \colorTwo{2} & \colorFou{4} & \colorOne{1} & \colorOne{1} & \colorFou{4} & \colorThr{3} & \colorThr{3} & \colorOne{1} & \colorOne{1} & \colorFou{4} & \colorThr{3} \\[-0.25em]
\colorFou{4} & \colorThr{3} & \colorThr{3} & \colorThr{3} & \colorTwo{2} & \colorTwo{2} & \colorTwo{2} & \colorFou{4} & \colorFou{4} & \colorTwo{2} & \colorTwo{2} & \colorTwo{2} & \colorOne{1} & \colorOne{1} & \colorOne{1} & \colorFou{4} & \colorFou{4} & \colorOne{1} & \colorOne{1} & \colorOne{1} & \colorThr{3} & \colorThr{3} & \colorThr{3} & \colorFou{4} \\[-0.5em]
    \end{tabular}
    \caption{Plain changes $\plain{n}$ using distinct one-sided ribbons for $n=4$.
    }
    \label{fig:plain4}
\end{figure}


It may be tempting to dismiss plain changes as trivial mathematics. 
To put this urge to rest, consider the following anecdote.
In 1964, Ron Graham and A. J. Goldstein attempted to generate permutations by swaps, unaware of plain changes, and their result \cite{goldstein1964sequential} is nowhere near as elegant%
\footnote{An internal Bell Labs report \cite{goldstein1964computer} contains an implementation of \cite{goldstein1964sequential}, but the authors have not been able to obtain a copy of it.}.
In other words, plain changes isn't obvious, even if you are a brilliant mathematician searching for~it.


\section{Twisted Plain Changes}
\label{sec:signedPlain}

Now we present our greedy solution to generating signed permutations in Definition \ref{def:twistedPlainChanges}.
In the remainder of this section we will prove that the approach does indeed generate a twist Gray code for $\SIGNED{n}$.

\begin{definition} \label{def:twistedPlainChanges}
\emph{Twisted plain changes} $\twisted{n}$ starts at ${+}1 \, {+}2 \, \cdots \, {+}n \in \SIGNED{n}$ then greedily $2$-twists the largest possible value and if none extends the list then it $1$-twists the largest possible value.
More precisely, it is the result of Algorithm~\ref{alg:greedy}.
\end{definition}


\begin{algorithm}[h]
\centering
\caption{Greedy algorithm for generating twisted plain changes~$\twisted{n}$.}
\label{alg:greedy}
\begin{algorithmic}[1]
\Procedure{$\AlgGreedy{n}$}{} \Comment{Signed permutations are visited in $\twisted{n}$ order}
\State $T \leftarrow \twtLeft{n}, \twtRight{n}, \ldots, \twtLeft{2}, \twtRight{2}, \twtRight{1}, \twtLeft{1}, \; \twt{n}, \ldots, \twt{2}, \twt{1}$ \Comment{List of $3n$ different twists}
\State $i \leftarrow 1$ \Comment{1-based index into $T$ (e.g., $T[1]$ is $\twtLeft{n}$)}
\State $\pi \leftarrow {+}1 \ {+}2 \ \cdots \ {+}n$ \Comment{Starting signed permutation $\pi \in \SIGNED{n}$}
\State $\Visit{\pi}$ \Comment{Visit $\pi$ for the first and only time}
\State $S = \{\pi\}$ \Comment{Add $\pi$ to the visited set}
\While{$i \leq 3n$} \Comment{Index $i$ iterates through the $3n$ different twists}
  \State $\pi' \leftarrow T[i](\pi)$ \Comment{Apply the $i^{\text{th}}$ highest priority twist to create $\pi'$}
  \If{$\pi' \notin S$} \Comment{Check if $\pi'$ is a new signed permutation}
    \State $\pi \leftarrow \pi'$ \Comment{Update the current signed permutation $\pi$}
    \State $\Visit{\pi}$ \Comment{Visit $\pi$ for the first and only time}
    \State $S = S \cup \{\pi\}$ \Comment{Add $\pi$ to the visited set}
    \State $i \leftarrow 1$ \Comment{Reset the 1-based index into $T$}
  \EndIf
\EndWhile
\EndProcedure
\end{algorithmic}
\end{algorithm}

Figure~\ref{fig:twist4} shows the start of $\twisted{n}$ for $n=4$.
Note that the first $24$ entries are obtained by $2$-twists.
The result is a familiar zig-zag pattern, but with every ribbon turning over during each pass.
The $25$th entry is obtained by a $1$-twist.





\begin{figure}[h]
    \includegraphics[angle=270,width=\textwidth-2em]{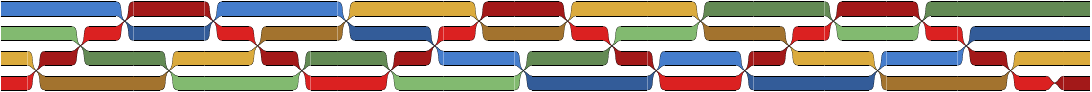} \raisebox{-1.5em}{{\large${\cdot}{\cdot}{\cdot}$}}
    \begin{tabular}{@{}*{23}{r@{\hspace{0.52em}}}r@{\hspace{0.52em}}|@{}r@{}}
\colorOne{$1$} & \colorOne{$1$} & \colorOne{$1$} & \colorFou{$\shortminus4$} & \colorFou{$\shortminus4$} & \colorOne{$1$} & \colorOne{$1$} & \colorOne{$1$} & \colorThr{$3$} & \colorThr{$3$} & \colorThr{$3$} & \colorFou{$\shortminus4$} & \colorFou{$\shortminus4$} & \colorThr{$3$} & \colorThr{$3$} & \colorThr{$3$} & \colorTwo{$\shortminus2$} & \colorTwo{$\shortminus2$} & \colorTwo{$\shortminus2$} & \colorFou{$\shortminus4$} & \colorFou{$\shortminus4$} & \colorTwo{$\shortminus2$} & \colorTwo{$\shortminus2$} & \colorTwo{$\shortminus2$} & \colorTwo{$\shortminus2$} \\[-0.20em]
\colorTwo{$2$} & \colorTwo{$2$} & \colorFou{$4$} & \colorOne{$\shortminus1$} & \colorOne{$\shortminus1$} & \colorFou{$4$} & \colorThr{$\shortminus3$} & \colorThr{$\shortminus3$} & \colorOne{$\shortminus1$} & \colorOne{$\shortminus1$} & \colorFou{$4$} & \colorThr{$\shortminus3$} & \colorThr{$\shortminus3$} & \colorFou{$4$} & \colorTwo{$2$} & \colorTwo{$2$} & \colorThr{$\shortminus3$} & \colorThr{$\shortminus3$} & \colorFou{$4$} & \colorTwo{$2$} & \colorTwo{$2$} & \colorFou{$4$} & \colorOne{$\shortminus1$} & \colorOne{$\shortminus1$} & \colorOne{$\shortminus1$} \\[-0.20em]
\colorThr{$3$} & \colorFou{$\shortminus4$} & \colorTwo{$\shortminus2$} & \colorTwo{$\shortminus2$} & \colorThr{$3$} & \colorThr{$3$} & \colorFou{$\shortminus4$} & \colorTwo{$\shortminus2$} & \colorTwo{$\shortminus2$} & \colorFou{$\shortminus4$} & \colorOne{$1$} & \colorOne{$1$} & \colorTwo{$\shortminus2$} & \colorTwo{$\shortminus2$} & \colorFou{$\shortminus4$} & \colorOne{$1$} & \colorOne{$1$} & \colorFou{$\shortminus4$} & \colorThr{$3$} & \colorThr{$3$} & \colorOne{$1$} & \colorOne{$1$} & \colorFou{$\shortminus4$} & \colorThr{$3$} & \colorThr{$3$} \\[-0.20em]
\colorFou{$4$} & \colorThr{$\shortminus3$} & \colorThr{$\shortminus3$} & \colorThr{$\shortminus3$} & \colorTwo{$2$} & \colorTwo{$2$} & \colorTwo{$2$} & \colorFou{$4$} & \colorFou{$4$} & \colorTwo{$2$} & \colorTwo{$2$} & \colorTwo{$2$} & \colorOne{$\shortminus1$} & \colorOne{$\shortminus1$} & \colorOne{$\shortminus1$} & \colorFou{$4$} & \colorFou{$4$} & \colorOne{$\shortminus1$} & \colorOne{$\shortminus1$} & \colorOne{$\shortminus1$} & \colorThr{$\shortminus3$} & \colorThr{$\shortminus3$} & \colorThr{$\shortminus3$} & \colorFou{$4$} & \colorFou{$\shortminus4$} \\[-0.20em]
    \end{tabular}
    \caption{Twisted plain changes $\twisted{n}$ for $n=4$ up to its $25$th entry.}
    \label{fig:twist4}
\end{figure}

\subsection{$2$-Twisted Permutohedron}
\label{sec:signedPlain_permutohedron}

At this point it is helpful to compare the start of plain changes and twisted plain changes.
The \emph{permutohedron of order $n$} is a graph whose vertices are permutations $\PERMS{n}$ and whose edges join two permutations that differ by a swap.
Plain changes traces a Hamilton path in this graph, as illustrated in Figure \ref{fig:permuto4_plain}.

Now consider signing each vertex $p_1 p_2 \cdots p_n$ in the permutohedron as~follows:
\begin{equation} \label{eq:2twistedPermuto}
p_j = i \text{ is positive if and only if } i \equiv j \mod 2.
\end{equation}
In particular, the permutation $12 \cdots n$ is signed as ${+}1{+}2 \cdots {+}n$ due to the fact that odd values are in odd positions, and even values are in even positions. 
One way of interpreting \eqref{eq:2twistedPermuto} is that swapping a symbol changes its sign.
Thus, after this signing, the edges in the resulting graph model $2$-twists instead of swaps.
For this reason, we refer to the graph as a \emph{$2$-twisted permutohedron of order $n$}.

Since twisted plain changes prioritizes $2$-twists before $1$-twists, the reader should be able to conclude that $\twisted{n}$ starts by creating a Hamilton path in the $2$-twisted permutohedron.
This is illustrated in Figure \ref{fig:permuto4_twisted}.

\begin{figure}
    \centering
    \begin{subfigure}{0.49\textwidth}
        \includegraphics[width=0.8\textwidth]{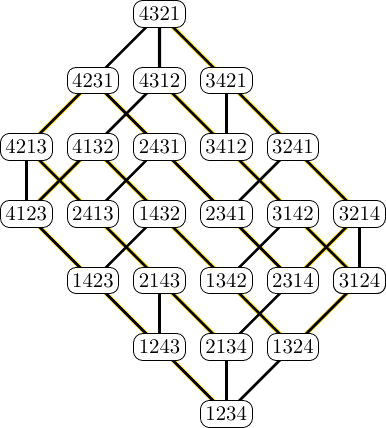}
        \caption{The Hamilton path from $1234 \in \PERMS{n}$ follows plain changes $\plain{4}$.}
        \label{fig:permuto4_plain}
    \end{subfigure}
    \hfill
    \begin{subfigure}{0.49\textwidth}
        \includegraphics[width=0.8\textwidth]{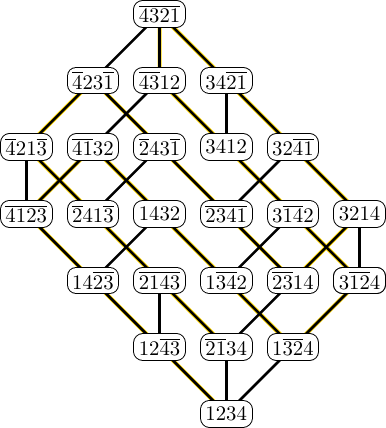}
        \caption{The Hamilton path from $1234 \in \SIGNED{n}$ follows twisted plain changes $\twisted{4}$.}
        \label{fig:permuto4_twisted}
    \end{subfigure}
    \caption{The permutohedron and the $2$-twisted permutohedron for $n=4$.}
    \label{fig:permuto4}
\end{figure}

\subsection{Global Structure}
\label{sec:signedPlain_global}

Our greedy approach can be verified to work for small $n$. 
To prove that it works for all $n$ we'll need to specify the global structure of the order that is created.

\begin{theorem} \label{thm:main}
Algorithm \ref{alg:greedy} creates a twist Gray code for signed permutations which starts at ${+}1 {+}2 {+}3 \cdots {+}n \in \SIGNED{n}$ and prioritizes $2$-twists of the largest value followed by $1$-twists of the largest value.
The order is cyclic as the last signed permutation is ${-}1 {+}2 {+}3 \cdots {+}n \in \SIGNED{n}$.
\end{theorem}
\begin{proof}
Since $2$-twists are prioritized before $1$-twists, the algorithm proceeds in the same manner as plain changes, except for the signs of the visited objects.
As a result, it generates sequences of $n!$ signed permutations using $2$-twists until a single $1$-twist is required.
One caveat is that the first signed permutation in a sequence alternates between having the underlying permutation of $1 2 3 4 \cdots n$ or $2 1 3 4 \cdots n$.
This is due to the fact that plain changes starts at $1 2 3 4 \cdots n$ and ends at $2 1 3 4 \cdots n$ and swaps $12$ to $21$ one time.
As a result, $12$ will be inverted while traversing every second sequence of length $n!$.
More specifically, the order generated by the algorithm appears in Figure \ref{fig:global}, with an example in Figure \ref{fig:path4}.
\qed


\end{proof}

\begin{figure}[h!]
\noindent
\small
\begin{tabular}{cccccccc}
$1 \, 2 \, 3 \cdots n{-}2 \, n{-}1 \, n$ & $\rightarrow$ & $1 \, 2 \, 3 \cdots n{-}2 \, \overline{n} \, \overline{n{-}1}$ & $\rightarrow$ & $\ldots$ & $\rightarrow$ & $\overline{2} \, \overline{1} \, 3 \cdots n{-}2 \, n{-}1 \, n$ & \\
& & & & & & $\downarrow$ & \\
$1 \, 2 \, 3 \cdots n{-}2 \, n{-}1 \, \overline{n}$ & $\leftarrow$ & $\cdots$ & $\leftarrow$ &  $\overline{2} \, \overline{1} \, 3 \cdots n{-}2 \, n \, \overline{n{-}1}$ & $\leftarrow$ & $\overline{2} \, \overline{1} \, 3 \cdots n{-}2 \, n{-}1 \, \overline{n}$ & \\
$\downarrow$ & & & & & & & \\
$1 \, 2 \, 3 \cdots n{-}2 \, \overline{n{-}1} \, \overline{n}$ & $\rightarrow$ & $1 \, 2 \, 3 \cdots n{-}2 \, n \, n{-}1$ & $\rightarrow$ & $\cdots$ & $\rightarrow$ & $\overline{2} \, \overline{1} \, 3 \cdots n{-}2 \, \overline{n{-}1} \, \overline{n}$ & \\
& & & & & & $\downarrow$ & \\
$1 \, 2 \, 3 \cdots n{-}2 \, \overline{n{-}1} \, n$ & $\leftarrow$ & $\cdots$ & $\leftarrow$ & $\overline{2} \, \overline{1} \, 3 \cdots n{-}2 \, \overline{n} \, n{-}1$ & $\leftarrow$ & $\overline{2} \, \overline{1} \, 3 \cdots n{-}2 \, \overline{n{-}1} \, n$ & \\
$\downarrow$ & & & & & & & \\
$1 \, 2 \, 3 \cdots \overline{n{-}2} \, \overline{n{-}1} \, n$ & $\rightarrow$ & $1 \, 2 \, 3 \cdots \overline{n{-}2} \, \overline{n} \, n{-}1$ & $\rightarrow$ & $\cdots$ & $\rightarrow$ & $\overline{2} \, \overline{1} \, 3 \cdots \overline{n{-}2} \, \overline{n{-}1} \, n$ & \\
& & & & & & $\downarrow$ & \\
& & & & & & $\vdots$ & \\
& & & & & & $\downarrow$ & \\
$\overline{1} \, 2 \, 3 \cdots n{-}2 \, n{-}1 \, n$ & $\leftarrow$ & $\cdots$ & $\leftarrow$  & $2 \, \overline{1} \, 3 \cdots n{-}2 \, \overline{n} \, \overline{n{-}1}$ & $\leftarrow$ & $2 \, \overline{1} \, 3 \cdots n{-}2 \, n{-}1 \, n$ & \\
\end{tabular}
\normalsize
\caption{The global structure of twisted plain changes.
Each row greedily applies $2$-twists to the largest possible symbol, thus following plain changes.
At the end of a row, no $2$-twist can be applied, and the down arrows greedily $1$-twist the largest possible symbol.
The rows alternate left-to-right and right-to-left (i.e., in boustrophedon order).
The leftmost column contains $1 2 \cdots n$ signed according to successive strings in the binary reflected Gray code.
The overall order is cyclic as a $1$-twist on value $1$ transforms the last entry into the first.}
\label{fig:global}
\end{figure}

\begin{figure}
    \centering
    \includegraphics[width=\textwidth]{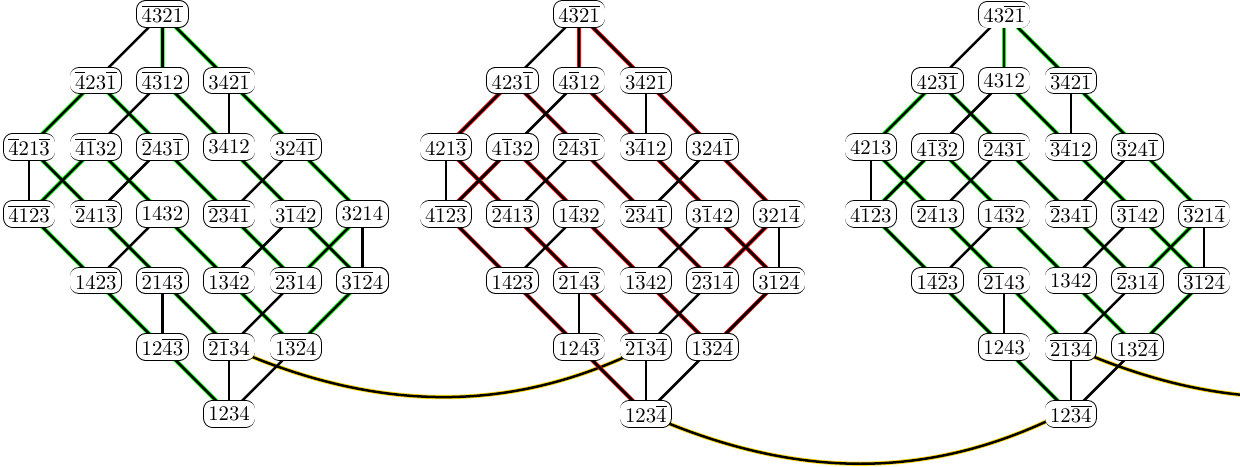}
    \caption{Illustrating the start of our twisted plain changes for $n=4$ traversing through the underlying flip graph.
    The straight edges denote every possible $2$-twists amongst the vertices shown, while the curved edges are the $1$-twists used amongst the vertices shown.
    The highlighted edges are used by the greedy algorithm, with green denoting subpaths starting at a signed $1 2 3 4 \cdots n$ vertex and red subpaths starting at signed $2 1 3 4 \cdots n$ vertex.}
    \label{fig:path4}
\end{figure}


The global structure described by Theorem \ref{thm:main} can be used as the foundation for a CAT algorithm.
We'll aim higher by developing a loopless algorithm across Sections \ref{sec:ruler}--\ref{sec:loopless}.
Besides offering optimal run-time performance (in a theoretical sense), the loopless algorithms help integrate our result with existing results.


\section{Ruler Sequences}
\label{sec:ruler}

In this section, we consider two general integer sequences, and how they relate to the Gray codes discussed up to this point.
As we'll see, these sequences help us understand our new greedy algorithm from Section \ref{sec:signedPlain}.
This work will also be central to the loopless generation of $\twisted{n}$ in Section~\ref{sec:loopless}.
The results in this section are illustrated in Table \ref{tab:orders}.

\subsection{Unsigned Sequences}
\label{sec:ruler_unsigned}

The \emph{ruler sequence} with bases $b_1, b_2, \ldots, b_n$ can be inductively defined as follows, where commas join sequences, and exponentiation denotes repetition.
\begin{align}
\ruler{b_1} &= 1^{b_1-1} = \overbrace{1, 1, \ldots, 1}^{b_1-1 \text{ copies}} \\
\ruler{b_1, b_2, \ldots, b_{i}} &= (\ruler{b_1, b_2, \ldots, b_{i-1}}, i)^{b_i-1}, \ruler{b_1, b_2, \ldots, b_{i-1}}
\end{align}
In other words, each successive base $b_i$ causes the pattern to be repeated $b_i$ times, with individual copies of $i$ interspersed between them.
Its name comes from the tick marks on rulers and tape measures, which follow the same pattern as the \emph{decimal ruler sequence}, $\ruler{10, 10, \ldots, 10}$.


The binary ruler sequence $\ruler{2,2,\ldots,2}$ gives the number of bits that change when counting in decimal.
More precisely, the sequence gives the length of the suffix of the form $011 \cdots 1$ that rolls over to $100 \cdots 0$ to create the next binary string.
For example, $\ruler{2,2,2} = 1,2,1,3,1,2,1$ which matches the lengths of suffixes that are complemented below.
\begin{equation}
00\overline{0}, 0\overline{01}, 01\overline{0},\overline{011}, 10\overline{0}, 1\overline{01}, 11\overline{0}, 111.
\end{equation}
Likewise, the \emph{upstairs factorial ruler sequence} $\ruler{1,2,\ldots,n}$ and \emph{downstairs factorial ruler sequence} $\ruler{n,n{-}1,\ldots,1}$ structure counting in the two most common factorial bases. 

\subsubsection{Gray Codes with Unsigned Ruler Change Sequences}
\label{sec:ruler_unsigned_Gray}

Ruler sequences also structure the changes many well-known Gray codes.
For example, the binary reflected Gray code complements individual bits according to the binary ruler function.
In other words, our standard numerical order of $n$-bit binary strings, and the binary reflected Gray code, have the same change sequence, but it is interpreted in different ways, with the former complementing entire suffixes and the latter complementing only the first bit of the suffix.
Similarly, the upstairs factorial ruler sequence provides the flip lengths in Zaks' order.
Corbett's order uses the downstairs factorial ruler sequence but with prefix rotation lengths of $n, 2, n{-}1, 3, \ldots, \lfloor\frac{n}{2}\rfloor$.


\subsection{Signed Sequences}
\label{sec:ruler_signed}

As discussed, the binary reflected Gray code uses the ruler sequence as its change sequence.
More specifically, the $j$th bit is complemented for each entry $j$ in the ruler sequence.
Suppose that we want a more precise change sequence: should the $j$th bit be incremented from $0$ to $1$, or decremented from $1$ to $0$?
Towards this goal, we introduce a signed version of the ruler sequence, in which positive and negative entries will denote increments and decrements.


The \emph{signed ruler sequence} is defined inductively as follows.
\begin{align}
\sruler{b_1} &= 1^{b_1-1} = \overbrace{1, 1, \ldots, 1}^{b_1-1 \text{ copies}} \\
\sruler{b_1, b_2, \ldots, b_{i}} &= 
    \begin{cases} 
        \label{eq:srulerInductive} 
        (s,i,\overline{s}^R,i)^{(b_i-1)/2},s & \text{if $b_i$ is odd} \\
        (s,i,\overline{s}^R,i)^{(b_i-2)/2},s,i,\overline{s} & \text{if $b_i$ is even}
    \end{cases} 
\end{align}
where $s = \sruler{b_1, b_2, \ldots, b_{i-1}}$, the overlines complement the sign of each entry, the $^{R}$ reverses the sequence, and the sequence ends with $s$ or $\overline{s}^R$ depending on the parity of $i$.
In other words, each successive base $b_i$ causes the pattern to be repeated $b_i$ times, with individual copies of $i$ interspersed between them, and every second copy being both complemented and reversed.
(It's worth noting that the unsigned ruler sequence is always a palindrome, so reversing every second recursive copy would not change the result.)


\subsubsection{Gray Codes with Signed Ruler Change Sequences}
\label{sec:ruler_signed_Gray}


The signed ruler sequence is helpful for properly understanding plain changes, as positive and negative entries correspond to swapping symbols in different directions.
More specifically, an sequence entry of ${-}j$ means move symbol $n-j+1$ to the left, and an entry of ${+}j$ means move symbol 

\begin{table}
\centering
\begin{tabular}{ccc|cc|ccc}
$\srulerName$ & BRGC & Binary & $\rulerName$ & Zaks' & $\srulerName$ & plain & downstairs \\ 
$(2,2,2,2)$ & $b_4 b_3 b_2 b_1$ & $b_4 b_3 b_2 b_1$ & $(1,2,3,4)$ & $p_1 p_2 p_3 p_4$ & $(4,3,2,1)$ & changes & Gray code \\ \hline
${+}1$ & $000\inc{0}$ & $000\overline{0}$  & $2$ & $\overleftrightarrow{12}34$ & ${+}1$ & $12\swapl{34}$ & $\inc{0}00$  \\
${+}2$ & $00\inc{0}1$ & $00\overline{01}$  & $3$ & $\overleftrightarrow{213}4$ & ${+}1$ & $1\swapl{24}3$ & $\inc{1}00$  \\
${-}1$ & $001\dec{1}$ & $001\overline{0}$  & $2$ & $\overleftrightarrow{31}24$ & ${+}1$ & $\swapl{14}23$ & $\inc{2}00$  \\
${+}3$ & $0\inc{0}10$ & $0\overline{011}$  & $3$ & $\overleftrightarrow{132}4$ & ${+}2$ & $41\swapl{23}$ & $3\inc{0}0$  \\
${+}1$ & $011\inc{0}$ & $010\overline{0}$  & $2$ & $\overleftrightarrow{23}14$ & ${-}1$ & $\swapr{41}32$ & $\dec{3}10$  \\
${-}2$ & $01\dec{1}1$ & $01\overline{01}$  & $4$ & $\overleftrightarrow{3214}$ & ${-}1$ & $1\swapr{43}2$ & $\dec{2}10$  \\
${-}1$ & $010\dec{1}$ & $011\overline{0}$  & $2$ & $\overleftrightarrow{41}23$ & ${-}1$ & $13\swapr{42}$ & $\dec{1}10$  \\
${+}4$ & $\inc{0}100$ & $\overline{0111}$  & $3$ & $\overleftrightarrow{142}3$ & ${+}2$ & $\swapl{13}24$ & $0\inc{1}0$  \\
${+}1$ & $110\inc{0}$ & $100\overline{0}$  & $2$ & $\overleftrightarrow{24}13$ & ${+}1$ & $31\swapl{24}$ & $\inc{0}20$  \\
${+}2$ & $11\inc{0}1$ & $10\overline{01}$  & $3$ & $\overleftrightarrow{421}3$ & ${+}1$ & $3\swapl{14}2$ & $\inc{1}20$  \\
${-}1$ & $111\dec{0}$ & $101\overline{0}$  & $2$ & $\overleftrightarrow{12}43$ & ${+}1$ & $\swapl{34}12$ & $\inc{2}20$  \\
${-}3$ & $1\dec{1}10$ & $1\overline{011}$  & $4$ & $\overleftrightarrow{2143}$ & ${+}3$ & $43\swapl{12}$ & $32\inc{0}$  \\
${+}1$ & $101\inc{0}$ & $110\overline{0}$  & $2$ & $\overleftrightarrow{34}12$ & ${-}1$ & $\swapr{43}21$ & $\dec{3}21$  \\
${-}2$ & $10\dec{1}1$ & $10\overline{01}$  & $3$ & $\overleftrightarrow{431}2$ & ${-}1$ & $3\swapr{42}1$ & $\dec{2}21$  \\
${-}1$ & $100\dec{1}$ & $101\overline{0}$  & $2$ & $\overleftrightarrow{13}42$ & ${-}1$ & $32\swapr{41}$ & $\dec{1}21$  \\
       &              &         & $3$ & $\overleftrightarrow{314}2$ & ${-}2$ & $\swapr{32}14$ & $0\dec{2}1$  \\
       &              &         & $2$ & $\overleftrightarrow{41}32$ & ${+}1$ & $23\swapl{14}$ & $\inc{0}11$  \\
       &              &         & $4$ & $\overleftrightarrow{1432}$ & ${+}1$ & $2\swapl{34}1$ & $\inc{1}11$  \\
       &              &         & $2$ & $\overleftrightarrow{23}41$ & ${+}1$ & $\swapl{24}31$ & $\inc{2}11$  \\
       &              &         & $3$ & $\overleftrightarrow{324}1$ & ${-}2$ & $42\swapr{31}$ & $3\dec{1}1$  \\
       &              &         & $2$ & $\overleftrightarrow{42}31$ & ${-}1$ & $\swapr{42}13$ & $\dec{3}01$  \\
       &              &         & $3$ & $\overleftrightarrow{243}1$ & ${-}1$ & $2\swapr{41}3$ & $\dec{2}01$  \\
       &              &         & $2$ & $\overleftrightarrow{34}21$ & ${-}1$ & $21\swapr{43}$ & $\dec{1}01$  \\
       &              &         &     & $4321$                      &        & $2134$         & $001$  \\
\end{tabular}
\caption{
Understanding counting and classic Gray codes using (un)signed ruler sequences.
Left: The binary ruler sequence is the change sequence for the binary reflected Gray code (BRGC) with the bit to complement given by the sequence value and its direction of change given by the sign.  
It is also the change sequence for standard counting in binary with the unsigned sequence value providing the length of suffix to complement.
Middle: The unsigned upstairs ruler sequence provides the prefix-reversal lengths (i.e., flip lengths) for Zaks' Gray code of permutations.
Right: The signed downstairs ruler sequence provides the symbol and direction to swap in plain changes.
It is also the change sequence for the corresponding reflected Gray code for downstairs strings (which are the inversion vectors for plain changes).
Our twisted plain change Gray code is generated in a similar manner using the signed factorial ruler sequence $\sruler{n,n{-}1,\ldots,2,1, \; 2,2,\ldots,2}$ with changes in the factorial portion of the base indicating $2$-twists, and changes in the binary portion of the base indicating $1$-twists.}
\label{tab:orders}
\end{table}

\section{Loopless Algorithms}
\label{sec:loopless}

While the greedy algorithm that generates $\twisted{n}$ from Section \ref{sec:combgen_greedy} is easy to describe, it is also highly inefficient.
More specifically, it requires exponential space, since all previously created objects must be remembered.
In this section, we show that the order can be generated efficiently without remembering what has previously been generated.
More specifically, our history-free implementation is loopless, and it uses the sequence-based framework described in Section \ref{sec:ruler} to generate each successive change to apply.

\subsection{Loopless Ruler Sequences}
\label{sec:loopless_ruler}

Algorithm \ref{alg:ruler} provides loopless algorithms for generating the (un)signed ruler sequences for any base.
The pseudocode is adapted from Knuth's loopless reflected mixed-radix Gray code Algorithm M \cite{knuth2013art}.

\begin{algorithm}[h]
\centering
\caption{Loopless algorithms for generating the (un)signed ruler sequence for bases $\mathbf{b} = b_1, b_2, \ldots, b_{n}$.
The $\mathbf{a}$ values store a mixed-radix word: $0 \leq a_i < b_i$.
The signed version generates the reflected mixed-radix Gray code in $\mathbf{a}$ with the $\mathbf{d}$ values providing $\pm 1$ directions of change. 
Focus pointers are stored in $\mathbf{f}$.
For example, if $\mathbf{b} = 3,2$ then $\AlgRulerS{\mathbf{b}}$ yields $+1,+1,2,-1,-1$ while creating the mixed-radix words $00, 10, 20, 21, 11, 10$ in $\mathbf{a}$.
Note: $b_i \geq 2$ for all $0 \leq i < n$.}
\label{alg:ruler}
\begin{minipage}[t]{0.45\textwidth}
\begin{algorithmic}[1]
\Procedure{$\AlgRulerU{\mathbf{b}}$}{}
\State $a_{1} \ a_{2} \ \cdots \ a_{n} \hspace{0.9em} \leftarrow 0 \ 0 \ \cdots \ 0$ \label{line:inita}
\State $f_{1} \ f_{2} \ \cdots \ f_{n+1} \leftarrow 1 \ 2 \ \cdots \ n{+}1$ \label{line:initf}
\State 
\While{$f_1 \leq n$}
  \State $j \leftarrow f_1$
  \State $f_1 \leftarrow 1$
  \State $a_j \leftarrow a_j + 1$
  \State $\Yield{j}$ \label{line:yield}
  \If{$a_j == b_j - 1$} 
    \State $d_j \leftarrow -d_j$
    \State $f_j \leftarrow f_{j+1}$
    \State $f_{j+1} \leftarrow j+1$
  \EndIf 
\EndWhile
\EndProcedure
\end{algorithmic}
\end{minipage}
\begin{minipage}[t]{0.45\textwidth}
\begin{algorithmic}[1]
\Procedure{$\AlgRulerS{\mathbf{b}}$}{}
\State $a_{1} \ a_{2} \ \cdots \ a_{n} \hspace{0.9em} \leftarrow 0 \ 0 \ \cdots \ 0$
\State $f_{1} \ f_{2} \ \cdots \ f_{n+1} \leftarrow 1 \ 2 \ \cdots \ n{+}1$
\State $d_{1} \ d_{2} \ \cdots \ d_{n} \hspace{0.9em} \leftarrow 1 \ 1 \ \cdots \ 1$
\While{$f_1 \leq n$}
  \State $j \leftarrow f_1$
  \State $f_1 \leftarrow 1$
  \State $a_j \leftarrow a_j + d_j$ 
  \State $\Yield{d_j \cdot j}$ 
  \If{$a_j \in \{0, b_j - 1\}$} 
    \State $a_j \leftarrow 0$
    \State $f_j \leftarrow f_{j+1}$
    \State $f_{j+1} \leftarrow j+1$
  \EndIf 
\EndWhile
\EndProcedure
\end{algorithmic}
\end{minipage}
\end{algorithm}

\subsection{Ruler-Based Gray Codes}
\label{sec:loopless_Gray}

The routines in Algorithm \ref{alg:Gray} are enhanced versions of those in Algorithm \ref{alg:ruler}.
More specifically, they have two additional inputs, $\mathbf{s}$ and $\mathbf{fns}$, which allow for a Gray code of objects to be created alongside the (un)signed ruler sequences.
To generate our twisted plain changes, we need to set the parameters as follows:
\begin{itemize}
    \item The bases $\mathbf{b}$ have length $2n$ are are equal to $n,n{-}1,\ldots,1, \ 2,2,\ldots,2$.
    Observe that this will generate a ruler sequence of length $n! \cdot 2^n - 1$ as desired.
    \item The starting object $\mathbf{s}$ is ${+}1 \; {+}2 \; \cdots \; {+}n$.
    \item The $\mathbf{fns}$ should be indexed using $1,2,\ldots,2n$ and ${-}1,{-}2,\ldots,{-}2n$.
    In particular, the $\mathbf{fns}$ at indices $\pm1, \pm2, \ldots, \pm n$ should modify the object by $2$-twisting successively smaller symbols to the left (positive) or right (negative),
    while the $\mathbf{fns}$ with larger absolute values should $1$-twist progressively smaller values (regardless of the sign).
\end{itemize}

\begin{algorithm}[h]
\centering
\caption{Algorithms for generating Gray codes using the (un)signed ruler sequence.
The $\mathbf{fns}$ modify object $\mathbf{s}$ and are indexed by the sequence.
For example, if $\mathbf{b} = 3,2$ then $\AlgGrayS{\mathbf{b}}$ creates the Gray code by starting with 
$\mathbf{s}$ and then modifies it using the $\mathbf{fns}$ with indices $+1,+1,2,-1,-1$.
The overall algorithm is loopless if each function runs in worst-case $\bigO{1}$-time.}
\label{alg:Gray}
\begin{minipage}[t]{0.45\textwidth}
\begin{algorithmic}[1]
\Procedure{$\AlgGrayU{\mathbf{b}, \mathbf{s}, \mathbf{fns}}$}{}
\State $a_{1} \ a_{2} \ \cdots \ a_{n} \hspace{0.9em} \leftarrow 0 \ 0 \ \cdots \ 0$ 
\State $f_{1} \ f_{2} \ \cdots \ f_{n+1} \leftarrow 1 \ 2 \ \cdots \ n{+}1$ 
\State 
\State $\Yield{\mathbf{s}}$
\While{$f_1 \leq n$}
  \State $j \leftarrow f_1$
  \State $f_1 \leftarrow 1$
  \State $a_j \leftarrow a_j + 1$
  \State $\mathbf{s} \leftarrow \mathbf{fns}[j]$
  \State $\Yield{j, \mathbf{s}}$
  \If{$a_j == b_j - 1$} 
    \State $d_j \leftarrow -d_j$
    \State $f_j \leftarrow f_{j+1}$
    \State $f_{j+1} \leftarrow j+1$
  \EndIf 
\EndWhile
\EndProcedure
\end{algorithmic}
\end{minipage}
\begin{minipage}[t]{0.45\textwidth}
\begin{algorithmic}[1]
\Procedure{$\AlgGrayS{\mathbf{b}, \mathbf{s}, \mathbf{fns}}$}{}
\State $a_{1} \ a_{2} \ \cdots \ a_{n} \hspace{0.9em} \leftarrow 0 \ 0 \ \cdots \ 0$
\State $f_{1} \ f_{2} \ \cdots \ f_{n+1} \leftarrow 1 \ 2 \ \cdots \ n{+}1$
\State $d_{1} \ d_{2} \ \cdots \ d_{n} \hspace{0.9em} \leftarrow 1 \ 1 \ \cdots \ 1$
\State $\Yield{\mathbf{s}}$
\While{$f_1 \leq n$}
  \State $j \leftarrow f_1$
  \State $f_1 \leftarrow 1$
  \State $a_j \leftarrow a_j + d_j$ 
  \State $\mathbf{s} \leftarrow \mathbf{fns}[d_j \cdot j]$
  \State $\Yield{d_j \cdot j, \mathbf{s}}$
  \If{$a_j \in \{0, b_j - 1\}$} 
    \State $a_j \leftarrow 0$
    \State $f_j \leftarrow f_{j+1}$
    \State $f_{j+1} \leftarrow j+1$
  \EndIf 
\EndWhile
\EndProcedure
\end{algorithmic}
\end{minipage}
\end{algorithm}

\begin{theorem} \label{thm:loopless}
Twisted plain changes $\twisted{n}$ is generated by a loopless algorithm by adapting the standard mixed-radix ruler sequence generation algorithm.
\end{theorem}

A 1-page Python implementation of twisted plain changes appears in the Appendix, and it loosely follows to structure of Algorithms \ref{alg:ruler}--\ref{alg:Gray}.

\bibliographystyle{splncs04}
\bibliography{refs}

\appendix

\newpage

\section{Python Implementation}
\label{sec:Python}

\begin{minted}{python}
# Flip sign of value v in signed permutation p with unsigned inverse q
def flip(p, q, v):  # with 1-based indexing; p[0] and q[0] are ignored.
    p[q[v]] = -p[q[v]]
    return p, q

# 2-twists value v to the left / right using delta = -1 / delta = 1
def twist(p, q, v, delta):  # with 1-based indexing into both p and q.
    pos = q[v]           # Use inverse to get the position of value v.
    u = abs(p[pos+delta]) # Get value to the left or right of value v.
    p[pos], p[pos+delta] = -p[pos+delta], -p[pos]     # Twist u and v.
    q[v], q[u] = pos+delta, pos             # Update unsigned inverse.
    return p, q  # Return signed permutation and its unsigned inverse.

# Generate each signed permutation in worst-case O(1)-time.
def twisted(n):
  m = 2*n-1  # The mixed-radix bases are n, n-1, ..., 2, 1, 2,..., 2.
  bases = tuple(range(n,1,-1)) + (2,) * n     # but the 1 is omitted.
  word = [0] * m            # The mixed-radix word is initially 0^m.
  dirs = [1] * m            # Direction of change for digits in word.
  focus = list(range(m+1))  # Focus pointers select digits to change.
  flips =   [lambda p,q,v=v: flip(p,q,v)     for v in range(n,0,-1)]
  twistsL = [lambda p,q,v=v: twist(p,q,v,-1) for v in range(n,1,-1)]
  twistsR = [lambda p,q,v=v: twist(p,q,v, 1) for v in range(n,1,-1)]
  fns = [None] + twistsL + flips + flips[-1::-1] + twistsR[-1::-1]
  p = [None] + list(range(1,n+1))  # To use 1-based indexing we set
  q = [None] + list(range(1,n+1))  # and ignore p[0] = q[0] = None.
  yield p[1:]  # Pause the function and return signed permutation p.
  while focus[0] < m: # Continue if the digit to change is in word.
    index = focus[0]  # The index of the digit to change in word.
    focus[0] = 0      # Reset the first focus pointer.
    word[index] += dirs[index]  # Adjust the digit using its direction.
    change = dirs[index] * (index+1)  # Note: change can be negative.
    if word[index] == 0 or word[index] == bases[index]-1:  # If the
      focus[index] = focus[index+1] # mixed-radix word's digit is at
      focus[index+1] = index+1      # its min or max value, then update
      dirs[index] = -dirs[index]    # focus pointers, change direction.
    p, q = fns[change](p, q)  # Apply twist or flip encoded by change.
    yield p[1:]

# Demonstrating the use of our twisted function for n = 4.
for p in twisted(4): print(p)  # Print all 2^n n! signed permutations.
\end{minted}

\end{document}